\documentclass[12pt]{article}

\usepackage[margin=1in]{geometry}

\usepackage{amsmath,amssymb,amsfonts,amsthm,amstext}
\usepackage{tikz}

\usetikzlibrary{decorations.pathreplacing,decorations.markings}
\usepackage{subfig}
\usepackage{xcolor}
\usepackage{xspace}

\usepackage[colorlinks]{hyperref}

\tikzset{
  mid arrow/.style={postaction={decorate,decoration={
        markings,
        mark=at position .5 with {\arrow[#1]{stealth}}
      }}},
}

\newtheorem{theorem}{Theorem}

\newtheorem*{autocor1}{1D Autocorrelation Lemma}
\newtheorem*{autocor2}{2D Autocorrelation Lemma}
\newtheorem{lemma}[theorem]{Lemma}

\newtheorem{corollary}[theorem]{Corollary}

\newcommand{\aone}{\hyperref[lem:autocor]{1D Autocorrelation Lemma}\xspace}
\newcommand{\atwo}{\hyperref[lem:autocor2]{2D Autocorrelation Lemma}\xspace}
\theoremstyle{definition}

\newcommand{\he}{\tikz{\draw[thick] (0,0)--(0.2,0);\draw[fill=black] (0,0) circle (1pt);\draw[fill=black] (0.2,0) circle (1pt);}}
\newcommand{\ve}{\tikz{\draw[thick] (0,0)--(0,0.2);\draw[fill=black] (0,0) circle (1pt);\draw[fill=black] (0,0.2) circle (1pt);}}
\newcommand{\hee}{\tikz{\draw[thick] (0,0)--(0.2,0)--(0.4,0);\draw[fill=black] (0,0) circle (1pt);\draw[fill=black] (0.2,0) circle (1pt);\draw[fill=black] (0.4,0) circle (1pt);}}
\newcommand{\hve}{\tikz{\draw[thick] (0,0)--(0.2,0)--(0.2,-0.2);\draw[fill=black] (0,0) circle (1pt);\draw[fill=black] (0.2,0) circle (1pt);\draw[fill=black] (0.2,-0.2) circle (1pt);}}

\begin{document}
\title{Local gap threshold for frustration-free spin systems}
\author{David Gosset \footnote{Walter Burke Institute for Theoretical Physics and Institute for Quantum Information and Matter, California Institute of Technology. dngosset@gmail.com}\and Evgeny Mozgunov \footnote{Institute for Quantum Information and Matter, California Institute of Technology. mvjenia@gmail.com}}
\date{}
\maketitle
\begin{abstract}
We improve Knabe's spectral gap bound for frustration-free translation-invariant local Hamiltonians in 1D. The bound is based on a relationship between global and local gaps. The global gap is the spectral gap of a size-$m$ chain with periodic boundary conditions, while the local gap is that of a subchain of size $n<m$ with open boundary conditions. Knabe proved that if the local gap is larger than the threshold value $1/(n-1)$ for some $n>2$, then the global gap is lower bounded by a positive constant in the thermodynamic limit $m\rightarrow \infty$.  Here we improve the threshold to $\frac{6}{n(n+1)}$, which is better (smaller) for all $n>3$ and which is asymptotically optimal. As a corollary we establish a surprising fact about 1D translation-invariant frustration-free systems that are gapless in the thermodynamic limit: for any such system the spectral gap of a size-$n$ chain with open boundary conditions is upper bounded as $O(n^{-2})$. This contrasts with gapless frustrated systems where the gap can be $\Theta(n^{-1})$. It also limits the extent to which the area law is violated in these frustration-free systems, since it implies that the half-chain entanglement entropy is $O(1/\sqrt{\epsilon})$ as a function of spectral gap $\epsilon$. We extend our results to frustration-free systems on a 2D square lattice.

\end{abstract}

\section{Introduction}
The spectral gap of a quantum many-body system is the difference between its first excited and ground energies. Here we study the asymptotic scaling of the spectral gap in the thermodynamic limit in which system size diverges. This scaling has physical consequences because the ground state correlation length \cite{Has04a,Has04b} and entanglement entropy (at least for one-dimensional systems \cite{area, AKLV13}) are both upper bounded as $\tilde{O}(1/\epsilon)$ where $\epsilon$ is the spectral gap\footnote{The $\tilde{O}(\cdot)$ notation hides a polylogarithmic function of $\frac{1}{\epsilon}$ which is present in the entanglement entropy bound from reference \cite{AKLV13}.}. In critical systems the divergence of these quantities is constrained by the rate at which the gap closes. The scaling of the gap can also have algorithmic consequences, since in some cases it determines the time complexity of the quantum adiabatic algorithm \cite{adiabatic}.

The most basic distinction is between gapped and gapless systems\footnote{Here we consider a quantum many-body system described by a sequence of Hamiltonians $\{H_n\}$ indexed by system size $n$. The system is gapped if the spectral gap of $H_n$ is lower bounded by a positive constant independent of $n$; otherwise it is gapless.}. How can one determine whether or not a given quantum many-body system is gapped? One might compute the spectral gap for a range of system sizes, and then attempt to extrapolate to thermodynamic limit. Unfortunately, no version of this strategy (or any other strategy) can work in the general case. Indeed, it was recently shown that determining whether or not a translation-invariant two-dimensional system is gapped or gapless is undecidable \cite{CubittGap2015}. Even for one-dimensional systems we are not aware of any methods.

The gapped versus gapless question seems to be more approachable if one restricts to frustration-free systems.  A frustration-free local Hamiltonian, given as a sum $\sum_i H_i$ of local terms, has the property that each of its ground states is also in the ground space of each term $H_i$. By a constant energy shift and rescaling we may ensure that each term $H_i$ has smallest eigenvalue $0$ and $\|H_i\|\leq 1$. In fact, we can (and do) assume without loss of generality that each term $H_i$ is a projector \footnote{If $H_i$ is not a projector, we may replace it with the projector $\Pi_i$ orthogonal to its null space. It is not hard to see that the new Hamiltonian $H'=\sum_{i} \Pi_i$ has the same zero energy ground space as $H$ and that the spectral gaps $\epsilon$ and $\epsilon'$ of these Hamiltonians satisfy $a\epsilon' \leq \epsilon \leq \epsilon'$, where $a$ lower bounds the smallest non zero eigenvalue of each term $H_i$.}, i.e.,  $H_i^2=H_i$. Frustration-freeness means that any ground state $|\psi\rangle$ of $H$ satisfies $H_i|\psi\rangle=0$ for all $i$. In the following we specialize to frustration-free translation-invariant 1D systems with nearest-neighbor interactions (later we discuss an extension to 2D systems). 

In this setting we are aware of two techniques for bounding the spectral gap in the thermodynamic limit. The Martingale method, due to Nachtergaele \cite{Nachtergaele1996}, establishes a system-size independent lower bound on the spectral gap as long as the ground space satisfies a certain condition (a relationship between the ground space projectors for overlapping contiguous regions of the chain).  In this paper we focus on the second technique, due to Knabe \cite{K88}, which is based on a relationship between the ``global" and ``local'' spectral gaps.  The global gap is the spectral gap of a Hamiltonian which describes a size-$m$ chain with periodic boundary conditions. The local gap is the spectral gap of a subchain of size $n<m$ with open boundary conditions. Knabe proved that if the local gap is larger than $1/(n-1)$ for some $n>2$, then the global gap is lower bounded by a positive constant in the thermodynamic limit $m\rightarrow \infty$.  This is an easy-to-use criterion for gappedness. In practice one can use exact numerical diagonalization to compute the spectral gap of the open boundary chain for small system sizes $n$. If for some $n$ one finds a value larger than $\frac{1}{n-1}$ then this establishes that the chain with periodic boundary conditions is gapped in the thermodynamic limit. 

Our main result is an improvement of the \textit{local gap threshold} which guarantees a gap in the thermodynamic limit, from $1/(n-1)$ to $\frac{6}{n(n+1)}$. This improvement is asymptotically optimal as the $O(n^{-2})$ scaling cannot be improved, see Section \ref{sec:mainresult}. We also show that the constant $6$ is close to optimal; in particular, it cannot be decreased below $\pi^2/2=4.93\ldots$. Since we lower the threshold our bound implies a gap in many cases where Knabe's does not.  On the other hand if one can establish that the spectral gap of a finite-size system exceeds Knabe's threshold then either result can be used to compute a constant lower bound on the gap in the thermodynamic limit. In that case one may obtain a better (larger) constant using Knabe's bound.

Our result elucidates a fundamental property of gapless frustration-free one-dimensional systems: roughly, that such systems have spectral gap upper bounded as $O(n^{-2})$. More precisely, if the system with periodic boundary conditions is gapless in the thermodynamic limit, then the spectral gap of the open boundary chain of size $n$ is at most $\frac{6}{n(n+1)}$. This constraint is specific to the frustration-free case; there are many examples of gapless frustrated Hamiltonians which have spectral gap which scales as $\Theta(n^{-1})$.  For example, critical systems which have a scaling limit described by a conformal field theory (such as the transverse field Ising model) have spectral gap which scales in this way.  

Reference \cite{BG15} provides a complete classification of translation-invariant frustration-free qubit chains with open boundary conditions, into gapped and gapless cases. Knabe's result was used to provide an upper bound of $1/(n-1)$ on the gap for the gapless cases, see Theorem 1 of \cite{BG15}. Using our bound in place of Knabe's in the proof, one obtains a strengthening of that theorem with the improved upper bound $\frac{6}{n(n+1)}$.

Knabe described an extension of his bound to two-dimensional systems on a hexagonal lattice, where he obtains a local gap threshold which scales inverse linearly with the diameter of a certain local region \cite{K88} (he was interested in applying the bound to the AKLT model \cite{AKLT87}). We also extend our results to two-dimensions where we find a local gap threshold that scales inverse quadratically with diameter. Here we consider systems with nearest-neighbor interactions on a square lattice. We chose the square lattice because it illustrates a new challenge which arises in two dimensions (see the second paragraph of Section \ref{sec:proof2}) and we do not know how to prove a bound for general lattices. For us the global gap is the spectral gap of a system with periodic boundary conditions in both spatial directions, while the local gap is the spectral gap of a smaller-sized region which for technical reasons is chosen to be an $n\times n$  patch with two rough and two smooth boundaries (see Figure \ref{fig:orientation}). We prove that if the local gap is greater than $8/n^2$ for some $n$ then the global gap is lower bounded by a positive constant independent of system size.  Note that in two dimensions one could imagine defining a local gap for any connected subgraph of the lattice. It is an open question whether the inverse quadratic (in the diameter) local gap threshold that we obtain for the patch can be generalized to any choice of subgraph.

It is an open question whether or not there is a local gap threshold for translation-invariant systems which may be frustrated. If such a threshold exists in 1D it must be $\Omega(\frac{1}{n})$ due to known examples such as the transverse Ising chain. We also note that, to the best of our knowledge, the existence of a local gap threshold for 2D frustrated systems would not contradict the recent undecidability result \cite{CubittGap2015}. A local gap threshold would allow one to certify that a system is gapped, but it would not provide an algorithm which decides if a system is gapped or gapless.
  
\subsection*{Entanglement versus gap}
Let $|\psi\rangle$ be the ground state of a local Hamiltonian which describes a 1D chain of $n$ qudits and consider its entanglement entropy $S(A)$ for some contiguous region $A\subseteq [n]$.  The area law \cite{area} states that if the system is gapped then $S(A)$ is upper bounded by a constant. Versions of this statement proven in references \cite{area,AKLV13} also provide an upper bound on the entanglement entropy as a function of the gap $\epsilon_n$. The strongest known upper bound is \cite{AKLV13}
\begin{equation}
S(A)=\tilde{O}\left(\epsilon_n^{-1}\right).
\label{eq:arealaw}
\end{equation}
One can ask whether this is tight. In other words what is the maximal violation of the area law in a gapless system? References \cite{evgap, gse,movshor} construct examples of systems with large entanglement entropy as a function of gap. The strongest known lower bound is given in reference \cite{evgap} which provides an example where $S(A)= \tilde{\Omega}(\epsilon_n^{-1/4})$.

Our result directly translates into an upper bound on entanglement entropy as a function of gap for a large class of gapless frustration-free systems.  More precisely, consider a translation-invariant frustration-free chain of qudits with periodic boundary conditions which is gapless in the thermodynamic limit.  Let $\epsilon_n$ be the spectral gap of the corresponding chain of size $n$ with open boundary conditions.  The entanglement entropy $S(A)$ for any ground state $|\psi\rangle$ and region $A\subseteq [n]$ is trivially upper bounded as
\begin{equation}
S(A) \leq \log(d^n)= O\left(\frac{1}{\sqrt{\epsilon_n}}\right) 
\label{eq:vs}
\end{equation}
where we used our result which states that $\epsilon_n\leq \frac{6}{n(n+1)}$. We may compare \eqref{eq:vs} with the area law from reference \cite{AKLV13} which gives the weaker bound \eqref{eq:arealaw} but applies more generally to systems which are gapped in the thermodynamic limit, systems which may be frustrated, and systems without translation invariance. 

It is natural to ask if equation \eqref{eq:vs} extends to frustration-free systems which are gapped in the thermodynamic limit. If true this would be a stronger version of the area law for such systems.  There is an argument in favor of that scenario. It is shown in reference \cite{GH15} that correlation length $\xi$ in frustration-free systems is upper bounded as 
\begin{equation}
\xi=O\left(\frac{1}{\sqrt{\epsilon_n}}\right),
\label{eq:cor}
\end{equation}
a square-root improvement over the best possible bound for the general (frustrated) case. A non-rigorous physics picture suggests that the entanglement across a cut should be contained within a region around the cut of size $O(\xi)$ where $\xi$ is the correlation length. This suggests that $d^n$ in \eqref{eq:vs} should be replaced by $d^\xi$, which along with \eqref{eq:cor} suggests that \eqref{eq:vs} should extend to gapped systems.\footnote{This argument has not been made rigorous, and it was only recently shown that exponential decay of correlations implies an area law at all \cite{decayarea}.  On the other hand for frustrated systems the same argument suggests the bound $S(A)=O(1/\epsilon_n)$ which was proven (up to a polylog factor) in reference \cite{AKLV13}.}

Equations \eqref{eq:vs} and \eqref{eq:cor} are both specific to the frustration-free case and represent square-root tightenings of bounds which hold more generally. Is there a unifying explanation for this square-rootiness? At present we do not have a satisfying answer to this question. However it might be interesting to  explore the connnection between these results and the notion of spectral gap amplification introduced in reference \cite{gapamp}. That work shows that one can associate a frustrated Hamiltonian $H$ to any frustration-free Hamiltonian $H_0$ such that (a) The ground space of $H_0$ is an (excited) eigenspace of $H$, and  (b) the corresponding eigenvalue gap in the spectrum of $H$ is $\sqrt{\epsilon}$, where $\epsilon$ is the spectral gap of $H_0$. 

The rest of this paper is organized as follows. In Section \ref{sec:mainresult} we state our bound for one-dimensional systems. In Section \ref{sec:strategy} we describe our proof strategy and how it relates to Knabe's original technique. We prove the bound in Section \ref{sec:proof1}.  In Section \ref{sec:proof2} we state and prove our bound for two-dimensional systems. The Appendix contains proofs of two technical lemmas.

\section{Spectral gap bound for one-dimensional systems}\label{sec:mainresult}
In this Section we state and prove our spectral gap bound for frustration-free chains. In the following we work in the $m$-qudit Hilbert space $(\mathbb{C}^d)^{\otimes m}$. Here $d$ is the local qudit dimension which is arbitrary. Consider a one-dimensional chain with periodic boundary conditions described by the Hamiltonian
\[
H^{\circ}_m=\sum_{i=1}^{m-1} h_{i,i+1}+h_{m,1}.
\]
Each term $h_{i,i+1}$ is a projector (i.e., $h_{i,i+1}^2=h_{i,i+1}$) which acts nontrivially on two qudits $i,i+1$ and as the identity on all other qubits. We assume translation invariance: the action of $h_{i,i+1}$ on qudits $i,i+1$  is described by a $d^2\times d^2$ matrix which does not depend on $i$.

We also define the $n$-qudit open boundary chain 
\[
H_n=\sum_{i=1}^{n-1}h_{i,i+1}.
\]

We are interested in the case where $H_m^{\circ}$ is frustration-free, meaning that it has a nonempty null space. In this case $H_n$ is also frustration-free for all $n<m$.  We write $\epsilon^{\circ}_m$ and $\epsilon_n$ for the smallest non zero eigenvalue of $H^{\circ}_m$ and $H_n$  respectively.  Knabe \cite{K88} proved the following theorem which relates these ``global" and ``local" spectral gaps.
\begin{theorem}[Knabe \cite{K88}]
Let $n>2$ and $m>n$. Then
\[
\epsilon^{\circ}_m \geq \left(\frac{n-1}{n-2}\right)\left(\epsilon_n-\frac{1}{n-1}\right).
\]
\label{thm:knabe}
\end{theorem}
Knabe's theorem directly implies the following corollary, which can be used to establish that the periodic chain $H_m^{\circ}$ is gapped in the thermodynamic limit, using information about the gap of a finite size open boundary chain.
\begin{corollary}[Knabe \cite{K88}]
Suppose there exists a positive integer $n>2$ such that $\epsilon_n>\frac{1}{n-1}$. Then there exists a positive constant $c$ such that $\epsilon^{\circ}_m\geq c$ for all $m$.
\end{corollary}
The contrapositive of this statement is an upper bound on the gap for gapless frustration-free chains: if $H_m^{\circ}$ is not gapped in the thermodynamic limit, then $\epsilon_n\leq \frac{1}{n-1}$  for all $n>2$. 

Here we improve the local gap threshold $\frac{1}{n-1}$ which guarantees a gap in the thermodynamic limit, to $\frac{6}{n(n+1)}$.
\begin{theorem}
Let $n>2$ and $m>2n$. Then
\[
\epsilon^{\circ}_m \geq \frac{5}{6}\left(\frac{n^2+n}{n^2-4}\right)\left(\epsilon_n-\frac{6}{n(n+1)}\right).
 \]
\label{thm:per}
\end{theorem}
We describe the proof strategy in Section \ref{sec:strategy} and then we give the proof in Section \ref{sec:proof1}. Before proceeding we give an example that illustrates how tight our bound is. The ferromagnetic Heisenberg model is a qubit chain with nearest-neighbor interaction 
\[
h=\frac{1}{2}\left(|01\rangle-|10\rangle\right)\left(\langle01|-\langle10|\right).
\]
In this example the periodic chain is gapless in the thermodynamic limit. The gap of the open boundary chain of size $n$ is exactly $\epsilon_n=1-\cos(\pi/n)$ which is $\sim\frac{\pi^2}{2n^2}$ for large $n$. This shows that the threshold $\frac{6}{n(n+1)}$ has the best possible scaling $\Theta(\frac{1}{n^2})$, and that, while it might be possible to slightly decrease the constant $6$, it cannot be decreased below $\pi^2/2=4.93...$.

\subsection{Proof strategy}\label{sec:strategy}
Suppose we consider the $m$-qudit chain $H_m^{\circ}$ but restrict our attention to a subchain consisting of $n$ qudits $\{k,k+1,\ldots, n+k-1\}$ for some $k\in [m]$.  The part of the Hamiltonian which describes this subchain is
\[
A_{n,k}=\sum_{i=k}^{n-2+k} h_{i,i+1}.
\]
(Here and throughout this Section all indices are to be read modulo $m$.) It acts on the qudits $\{k,k+1,\ldots, n+k-1\}$  as the open boundary chain $H_n$ and as the identity on all other qudits. 

Knabe's proof of Theorem \ref{thm:knabe}\cite{K88} is based on comparing the operators $(H_m^{\circ})^2$ and $\sum_{k=1}^{m}A^2_{n,k}$. Expanding the squares as sums of operators $h_{i,i+1}h_{j,j+1}$ one obtains an inequality
\begin{equation}
(H_m^{\circ})^2+\beta H_m^{\circ}\geq \alpha \sum_{k=1}^{m}A^2_{n,k}
\label{eq:opin}
\end{equation}
where $\alpha=\beta=\frac{1}{n-2}$ are positive coefficients. Since the smallest non zero eigenvalue of $A_{n,k}$ is $\epsilon_n$, we have $A^2_{n,k}\geq \epsilon_n A_{n,k}$ and therefore
\[
\sum_{k=1}^{m}A^2_{n,k}\geq \epsilon_n \sum_{k=1}^m A_{n,k}=\epsilon_n\gamma H_m^{\circ}
\]
where $\gamma={n-1}$ is another positive coefficient. Plugging this into \eqref{eq:opin} and rearranging gives
\[
(H_m^{\circ})^2\geq \alpha\gamma \left(\epsilon_n-\frac{\beta}{\alpha \gamma}\right)H_m^{\circ}
\]
which establishes that $\epsilon^{\circ}_m\geq \alpha\gamma \left(\epsilon_n-\frac{\beta}{\alpha \gamma}\right)$. Substituting the values of $\alpha,\beta, \gamma$ quoted above one obtains the bound from Theorem \ref{thm:knabe}.

The central difference between our proof and Knabe's is that we work with deformed versions of the subchain operators $A_{n,k}$. Specifically, let $\{c_0,c_1,\ldots, c_{n-2}\}$ be a set of real numbers which satisfy the following three conditions
\begin{align}
 \text{(\textbf{Positive})} & \qquad c_j > 0 \qquad \quad &0 \leq j \leq n-2\label{eq:cond1}\\
\text{(\textbf{Nondecreasing up to  midpoint})} & \qquad c_j \geq c_{j-1} \qquad \quad  &1 \leq j \leq \frac{n-2}{2}\label{eq:cond2}\\
\text{(\textbf{Symmetric about midpoint})} & \qquad c_j =c_{(n-2)-j} \qquad \quad  &0 \leq j \leq \frac{n-2}{2}.\label{eq:cond3}
\end{align}
Define operators 
\begin{equation}
B_{n,k}=\sum_{i=k}^{n-2+k}c_{i-k}h_{i,i+1} \qquad k\in \{1,\ldots,m\}.
\label{eq:defB}
\end{equation}
Our proof is based on comparing the operators $(H_m^{\circ})^2$ and $\sum_{k=1}^{m}B^2_{n,k}$. Many of the steps described above can be made to work except that now the coefficients $\alpha, \beta$ and $\gamma$ depend on $\{c_0,c_1,\ldots, c_{n-2}\}$. A suitable choice of the $\{c_j\}$ gives the desired bound. 

The idea of deforming the subchain operators in this way was inspired by a calculation due to Alexei Kitaev which establishes an analogous local gap bound for one-dimensional matrices (e.g., tridiagonal matrices) \cite{Kit}. 

One challenge with adapting Knabe's proof is that while Knabe uses the operator inequality $A^2_{n,k}\geq \epsilon_n A_{n,k}$, the gap of $B_{n,k}$ is not related to $\epsilon_n$ in a simple way. It is possible to obtain a lower bound
\[
B_{n,k}^2 \geq \epsilon_n \left(\min_{0\leq j\leq n-2} {c_j}\right) B_{n,k}
\] 
but it turns out that this is not strong enough for our purposes. Happily, we find that the following Lemma can be used in place of a stronger lower bound on the eigenvalue gap of $B_{n,k}$. The proof of the Lemma uses the translation invariance of the periodic chain.

\begin{lemma}
Consider the eigenspace of $H^{\circ}_m$ with eigenvalue $\epsilon^{\circ}_m$. There exists a normalized state $|\phi\rangle$ in this eigenspace which satisfies
\begin{equation}
\langle \phi |B_{n,k}^2 |\phi\rangle \geq \epsilon_n \left(\frac{1}{n-1}\sum_{j=0}^{n-2}c_j\right) \langle \phi |B_{n,k}|\phi\rangle
\label{eq:B2bnd}
\end{equation}
for each $k=1,\ldots,m$.
\label{lem:cnorm}
\end{lemma}
\begin{proof}
Write $T$ for the space translation operator, that it, the unitary operator which cyclically translates the $m$ qudits by one, i.e., 
\[
T^\dagger h_{i,i+1} T=h_{i+1,i+2} \qquad i=1,\ldots,m.
\]
Since $[T,H^{\circ}_m]=0$, we may simultaneously diagonalize $T$ and  $H^{\circ}_m$. Note that all eigenvalues of $T$ are complex phases $e^{i\theta}$. So we may choose $|\phi\rangle$ such that $H^{\circ}_m|\phi\rangle=\epsilon_m^{\circ} |\phi\rangle$  and $T|\phi\rangle=e^{i\theta}|\phi\rangle$ for some $\theta\in \mathbb{R}$. Using this fact we see that the energy of $|\phi\rangle$  is distributed evenly among the terms 
\begin{equation}
\langle \phi| h_{i,i+1}|\phi\rangle =\langle \phi| T^{\dagger{i-j}} h_{j,j+1} T^{{i-j}} |\phi\rangle =\langle \phi| h_{j,j+1} |\phi\rangle \qquad i,j\in  \{1,\ldots,m\}.
\label{eq:energydist}
\end{equation}

Now let $G_{n,k}$ denote the projector onto the nullspace of $B_{n,k}$, and write $G^{\perp}_{n,k}=I-G_{n,k}$. If $ G^{\perp}_{n,k}|\phi\rangle=0$ then \eqref{eq:B2bnd} clearly holds since then $B_{n,k}^2|\phi\rangle=B_{n,k}|\phi\rangle=0$. Next suppose $ G^{\perp}_{n,k}|\phi\rangle\neq0$ and define
\begin{equation}
|\widehat{\phi}\rangle=\frac{1}{\big\|G^{\perp}_{n,k}|\phi\rangle\big\|} G^{\perp}_{n,k}|\phi\rangle.
\label{eq:phihat}
\end{equation}
Using the fact that $B_{n,k}=B_{n,k}G^{\perp}_{n,k}=G_{n,k}^\perp B_{n,k}$ we get
\begin{align}
\langle \phi |B_{n,k}^2 |\phi\rangle& =\langle \widehat{\phi}|B_{n,k}^2 |\widehat{\phi}\rangle \langle \phi|G^{\perp}_{n,k}|\phi\rangle\\
&\geq \left(\langle \widehat{\phi}|B_{n,k} |\widehat{\phi}\rangle\right)^2 \langle \phi|G^{\perp}_{n,k}|\phi\rangle\\
& =\langle \widehat{\phi}|B_{n,k} |\widehat{\phi}\rangle \langle \phi|B_{n,k} |\phi\rangle.\label{eq:Bsq}
\end{align}
Now 
\begin{equation}
\langle \widehat{\phi}|B_{n,k} |\widehat{\phi}\rangle =\sum_{j=k}^{n-2+k} c_{j-k} \langle \widehat{\phi}|h_{j,j+1}|\widehat{\phi}\rangle.
\label{eq:phi2}
\end{equation}
Using \eqref{eq:energydist} we see that $ \langle \widehat{\phi}|h_{i,i+1}|\widehat{\phi}\rangle $ takes the same value for all $i\in \{k,k+1,\ldots,n-2+k\}$, that is,
\begin{equation}
 \langle \widehat{\phi}|h_{i,i+1}|\widehat{\phi}\rangle =\frac{\langle \phi|G^{\perp}_{n,k}h_{i,i+1}G^{\perp}_{n,k}|\phi\rangle}{  \langle \phi|G^{\perp}_{n,k}|\phi\rangle}=\frac{\langle \phi|h_{i,i+1}|\phi\rangle}{  \langle \phi|G^{\perp}_{n,k}|\phi\rangle}=\frac{\langle \phi|h_{j,j+1}|\phi\rangle}{\langle \phi|G^{\perp}_{n,k}|\phi\rangle}=  \langle \widehat{\phi}|h_{j,j+1}|\widehat{\phi}\rangle 
\label{eq:sym}
\end{equation}
for all $i,j\in \{k,\ldots, n-2+k\}$. Using this fact in \eqref{eq:phi2} gives
\begin{equation}
\langle \widehat{\phi}|B_{n,k} |\widehat{\phi}\rangle =\left(\frac{1}{n-1}\sum_{j=0}^{n-2} c_j\right) \langle \widehat{\phi}|\sum_{i=k}^{n-2+k} h_{i,i+1}|\widehat{\phi}\rangle
\geq \left(\frac{1}{n-1}\sum_{j=0}^{n-2} c_{j}\right) \epsilon_n
\label{eq:phihat2}
\end{equation}
where in the last inequality we used the fact that $|\widehat{\phi}\rangle$ is a normalized state orthogonal to the nullspace of the operator $\sum_{i=k}^{n-2+k} h_{i,i+1}$ and therefore has energy lower bounded by its smallest non zero eigenvalue, which is $\epsilon_n$. Plugging \eqref{eq:phihat2} into \eqref{eq:Bsq} completes the proof.
\end{proof}

\subsection{Proof of Theorem \ref{thm:per}}\label{sec:proof1}
\begin{proof}
Let $\{c_j\}$ be a set of numbers which satisfy the conditions \eqref{eq:cond1}-\eqref{eq:cond3}, and recall the definition of the deformed subchain operator $B_{n,k}$ from equation \eqref{eq:defB}. We begin by taking squares 
\begin{align}
(H_m^{\circ})^{2} & =H_m^{\circ}+\sum_{i=1}^{m}\left(h_{i,i+1}h_{i+1,i+2}+h_{i+1,i+2}h_{i,i+1}\right)+\sum_{\substack{i,j\in[m]\\
|i-j|\geq2
}
}h_{i,i+1}h_{j,j+1}\label{eq:H_squared}\\
B_{n,k}^{2} & =\sum_{i=k}^{n-2+k}c_{i-k}^2h_{i,i+1}+\sum_{i=k}^{n-3+k}c_{i-k}c_{i+1-k}\left(h_{i,i+1}h_{i+1,i+2}+h_{i+1,i+2}h_{i,i+1}\right)\nonumber\\
&+\sum_{\substack{i,j\in\{k,...,n-2+k\}\\
|i-j|\geq2
}
}c_{i-k}c_{j-k}h_{i,i+1}h_{j,j+1}.\label{eq:Bnksq}
\end{align}

Next we evaluate $\sum_{k=1}^m B_{n,k}^2$. It will be convenient to define 
\[
d(i,j)=\min\{|i-j|,m-|i-j|\} \qquad i,j\in [m]
\] 
which is the minimum distance between vertices $i$ and $j$ on the $m$-cycle. We obtain

\begin{align}
\sum_{k=1}^{m} B_{n,k}^2&= \left(\sum_{j=0}^{n-2}c_j^2\right)H_m^{\circ} +\left(\sum_{j=0}^{n-3} c_jc_{j+1}\right)\sum_{i=1}^{m}\left(h_{i,i+1}h_{i+1,i+2}+h_{i+1,i+2}h_{i,i+1}\right)\nonumber\\
&+\sum_{\substack{i,j\in\{m\}\\
n-2\geq d(i,j)\geq2
}
}h_{i,i+1}h_{j,j+1}\left(\sum_{r=0}^{n-2-d(i,j)} c_r c_{r+d(i,j)}\right).
\label{eq:sumsq}
\end{align}
It is clear that the first two terms above are obtained by summing the first two terms in \eqref{eq:Bnksq} over $k$. In computing the third term we use the hypothesis of the theorem that $m>2n$ (otherwise one would have to include terms of the form $c_r c_{r+m-d(i,j)}$ in \eqref{eq:sumsq}).

Now define
\[
\alpha=\left(\sum_{j=0}^{n-3}c_j c_{j+1}\right)^{-1} \quad \text{and} \quad \beta=\alpha\left(\sum_{j=0}^{n-2}c_j^2-\sum_{j=0}^{n-3}c_j c_{j+1}\right).
\]
Using \eqref{eq:H_squared} and \eqref{eq:sumsq} and the fact that $h_{i,i+1}h_{j,j+1}\geq 0$ whenever $d(i,j)\geq 2$ we obtain
\begin{equation}
(H_m^{\circ})^2-\alpha\sum_{k=1}^{m} B_{n,k}^2+\beta H_m^{\circ} \geq \sum_{\substack{i,j\in\{m\}\\
n-2\geq d(i,j)\geq2
}
}h_{i,i+1}h_{j,j+1}\left(1-\alpha \sum_{r=0}^{n-2-d(i,j)} c_r c_{r+d(i,j)}\right).
\label{eq:positive}
\end{equation}
We see that the right-hand side is positive semidefinite as long as the coefficients $\{c_i\}$ satisfy
\begin{equation}
\sum_{j=0}^{n-3}c_j c_{j+1}-\sum_{j=0}^{n-x-2}c_j c_{j+x}\geq 0 \qquad  x=1,2,\ldots,n-2.
\label{eq:c_constraints}
\end{equation}
However this condition is guaranteed to hold for any choice of the coefficients $\{c_j\}$ which satisfy \eqref{eq:cond1}-\eqref{eq:cond3}:
\begin{autocor1}
Let $\{c_j: j\in\{0,1\ldots,n-2\}\}$ be any real numbers satisfying conditions \eqref{eq:cond1}-\eqref{eq:cond3}. Define 
\[
q(x)=\sum_{j=0}^{n-x-2}c_jc_{j+x} 
\]
Then $q(x)\geq q(x+1)$ for all $x=0,1,\ldots, n-3$.
\label{lem:autocor}
\end{autocor1}
A proof is provided in the Appendix. Using the Lemma we see that the right-hand side of \eqref{eq:positive} is positive semidefinite and therefore
\begin{equation}
(H_m^{\circ})^2+\beta H^\circ_m\geq \alpha\sum_{k=1}^{m} B_{n,k}^2.
\label{eq:posdef}
\end{equation}

Now let $|\phi\rangle$ be the state whose existence is guaranteed by Lemma \ref{lem:cnorm} and take the expectation value of equation \eqref{eq:posdef}:
\begin{equation}
(\epsilon_m^{\circ})^2+\beta \epsilon_m^{\circ} \geq \alpha \langle \phi | \sum_{k=1}^{m} B_{n,k}^2 |\phi\rangle 
\geq \alpha  \epsilon_n \left(\frac{1}{n-1}\sum_{j=0}^{n-2}c_j\right) \langle \phi |\sum_{k=1}^{m}B_{n,k}|\phi\rangle.
\end{equation}

Now using the fact that $\sum_{k=1}^{m}B_{n,k}=\left(\sum_{j=0}^{n-2}c_j\right)H^\circ_m$ and the fact that $\langle\phi |H^{\circ}_m|\phi\rangle=\epsilon^{\circ}_m$, we get
\begin{align}
(\epsilon_m^{\circ})^2+\beta \epsilon_m^{\circ} \geq \frac{\alpha}{n-1}\left(\sum_{j=0}^{n-2}c_j\right)^2 \epsilon_n\epsilon_m^{\circ}.
\end{align}
Dividing through by $\epsilon_m^{\circ}$ and rearranging gives
\begin{equation}
\epsilon_m^{\circ}\geq F(n)\left(\epsilon_n-G(n)\right)
\label{eq:prefactor}
\end{equation}
where 
\begin{align}
F(n)& =\frac{\alpha}{n-1}\left(\sum_{j=0}^{n-2}c_j\right)^2=\frac{1}{(n-1)\sum_{j=0}^{n-3}c_jc_{j+1}}\left(\sum_{j=0}^{n-2}c_j\right)^2\label{eq:F}\\
G(n)& =(n-1)\left(\frac{\sum_{j=0}^{n-2}c_j^2-\sum_{j=0}^{n-3}c_j c_{j+1}}{\left(\sum_{j=0}^{n-2}c_j\right)^2}\right).\label{eq:G}
\end{align}
We now choose the coefficients $\{c_j\}$ as follows
\begin{equation}
c_j= (n-1)+((n-2)j-j^2) \qquad j=0,\ldots, n-2.
\label{eq:cchoice}
\end{equation}
It is clear that the conditions \eqref{eq:cond1}-\eqref{eq:cond3} are satisfied. To complete the proof we show that this choice gives $F(n)=\frac{5}{6}\left(\frac{n^2+n}{n^2-4}\right)$ and $G(n)=\frac{6}{n(n+1)}$. In fact, one can use Lagrange multipliers to show that the choice \eqref{eq:cchoice} is optimal\footnote{ Note that $G(n)$ is invariant under a rescaling $c_j\rightarrow ac_j$, so we are free to choose the normalization so that $\sum_{j=0}^{n-2}c_j=1$. To find the optimum of $G(n)$ we minimize the numerator of \eqref{eq:G} subject to this constraint. The solution is equal to \eqref{eq:cchoice} (up to an irrelevant rescaling).} in the sense that it minimizes $G(n)$ for all $n>2$.

The sums which appear in \eqref{eq:F} and \eqref{eq:G} can be evaluated exactly:
\begin{align}
 \sum_{j=0}^{n-2} c_j &=\frac{n^3-n}{6} \label{eq:csum}\\
\sum_{j=0}^{n-2} c_j^2 &=\frac{n^5-n}{30}\label{eq:csq}\\
\sum_{j=0}^{n-3}c_jc_{j+1}&=\frac{n^5}{30}-\frac{n^3}{6}+\frac{2}{15}n.\label{eq:ccp1}
\end{align}
Plugging \eqref{eq:csum} and \eqref{eq:ccp1} into \eqref{eq:F} we obtain
\[
F(n)=\frac{1}{(n-1)}\frac{(n^3-n)^2}{(6/5n^5-6n^3+\frac{72}{15}n)}=\frac{5}{6}\frac{(n^2+n)(n^3-n)}{n^5-5n^3+4n}=\frac{5}{6}\left(\frac{n^2+n}{n^2-4}\right).
\]
Plugging \eqref{eq:csum}, \eqref{eq:csq}, and \eqref{eq:ccp1} into \eqref{eq:G} gives
\[
G(n)=(n-1)\frac{\frac{n^3-n}{6}}{\left(\frac{n^3-n}{6}\right)^2}=\frac{6}{n(n+1)},
\]
which completes the proof. 
\end{proof}

\section{Spectral gap bound for two-dimensional systems}\label{sec:proof2}

In this Section we establish a relationship between local and global gaps for translation-invariant frustration-free systems with nearest-neighbor interactions on a two-dimensional square lattice.  Here the global gap is the spectral gap of the system defined on the square lattice with periodic boundary conditions and the local gap is that of a connected subgraph which we refer to as a ``patch", defined below. 

Before stating our result we mention an essential difference between the 2D bound given here and the 1D bound from the previous Section. A chain has the property that all pairs of adjacent edges look the same, like $\hee$. When we square a 1D Hamiltonian $H_m^{\circ}$ and write it as a sum of terms $h_{i,i+1}h_{j,j+1}$ the only ones which are not positive semidefinite are those which come from pairs of adjacent edges, i.e., $j=i+1$.  In contrast on the 2D square lattice there are two types of pairs of adjacent edges, those which are collinear (e.g., $\hee$) and those which are not (e.g., $\hve$).  Note that for other lattices in two dimensions the number of types of pairs can be different; for example, in the hexagonal lattice which was considered by Knabe \cite{K88} there is only one type.  For the square lattice we are able to make the proof go through in the presence of this additional complication by choosing the shape of the patch in a special way.

 Let $\Lambda_m$ be an $m\times m$ two-dimensional square lattice with periodic boundary conditions in both directions, i.e., a torus. We consider a system of $m^2$ qudits which live at the vertices of $\Lambda_m$. The Hilbert space is $(\mathbb{C}^d)^{\otimes m^2}$ where $d$ is the qudit dimension which is arbitrary. We consider a translationally invariant frustration-free Hamiltonian with nearest-neighbor interactions, which we write as
\[
H_m^\mathrm{T}=\sum_{e\in \Lambda_m}h_e.
\]
Here the sum is over all edges $e$ in the lattice (here and in the following we do not distinguish between a graph and its edge set). The $\mathrm{T}$ superscript is for ``torus". For each edge $e$ there is a projector $h_e$ which acts nontrivially only on the two qudits at the endpoints of $e$. In general this projector will not be invariant under swapping these qudits, so we fix an orientation (direction) for each edge $e\in \Lambda_m$. On the torus there are two types of edges, horizontal and vertical ones. We assume translation invariance in each direction which means that all horizontal edges are directed the same way (left-to-right, say) and all vertical edges are directed the same way (bottom-to-top). For simplicity we assume that the system is isotropic so that the nontrivial action of $h_e$ is described by the same $d^2\times d^2$ Hermitian matrix, for all edges $e$. Below we discuss how our bound can be modified very slightly to handle the non isotropic case where the terms on vertical and horizontal edges may be different. We also assume frustration-freeness which means that the ground energy of $H_m^{\mathrm{T}}$ is zero.

Ultimately we aim to relate the spectral gap of the Hamiltonian $H_m^\mathrm{T}$ for the torus, and the spectral gap of a Hamiltonian defined on a smaller-sized patch. We choose the shape of the patch in the following way. For each positive even integer $n$, define a directed $n(n+2)$-vertex graph $P_n$ as follows. Start with an undirected grid graph with $n$ rows and $n+2$ columns and then remove all the vertical edges in the first and last columns. The result is an $n\times n$ square region with two smooth and two rough edges. We direct each horizontal edge left-to-right and each vertical edge bottom-to-top to obtain $P_n$. It will also be convenient to define another directed graph $Q_n$ which is obtained by taking $P_n$, rotating it by $\pi/2$, and then reassigning orientations of the edges so that horizontal edges are directed left-to-right and vertical edges bottom-to-top. The graphs $P_n$ and $Q_n$ for the case $n=4$ are shown in Figure \ref{fig:orientation}.

\begin{figure}[t]
\centering
\subfloat[]{
\begin{tikzpicture}
\foreach \x in {0,1,2,3,4}
\foreach \y in {0,1,2,3}
\draw[thick, postaction={mid arrow}] (\x,\y)--(\x+1,\y);

\foreach \x in {1,2,3,4}
\foreach \y in {0,1,2}
\draw[thick, postaction={mid arrow}] (\x,\y)--(\x,\y+1);
\foreach \x in {0,1,2,3,4,5}
\foreach \y in {0,1,2,3}
\draw[thick, fill] (\x,\y) circle (2pt);
\end{tikzpicture}
}
\hspace{2.5cm}
\subfloat []{
\begin{tikzpicture}
\foreach \y in {0,1,2,3,4}
\foreach \x in {0,1,2,3}
\draw[thick, postaction={mid arrow}] (\x,\y)--(\x,\y+1);

\foreach \y in {1,2,3,4}
\foreach \x in {0,1,2}
\draw[thick, postaction={mid arrow}] (\x,\y)--(\x+1,\y);
\foreach \y in {0,1,2,3,4,5}
\foreach \x in {0,1,2,3}
\draw[thick, fill] (\x,\y) circle (2pt);
\end{tikzpicture}
}
\caption{Directed graphs $P_n$ and $Q_n$ for $n=4$. \label{fig:orientation}}
\end{figure}
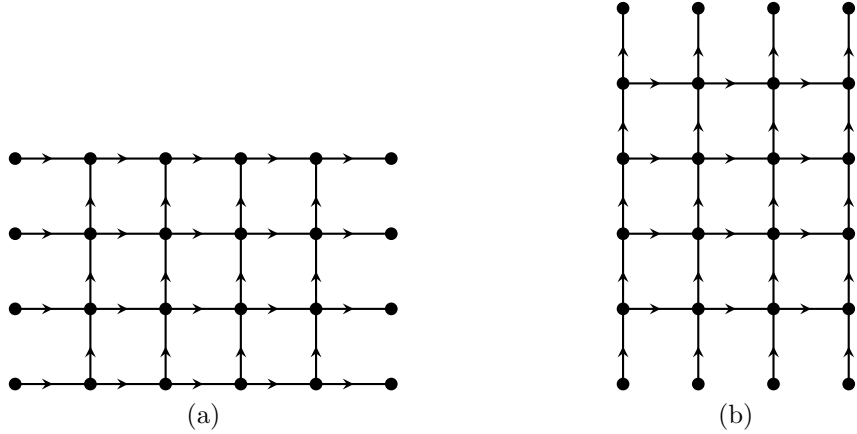

We define a patch Hamiltonian which acts on $n(n+2)$ qudits and has a term for each edge $e\in P_n$ 
\[
H_n^{\mathrm{P}}=\sum_{e\in P_n}h_e.
\]
We also define $H_n^{\mathrm{Q}}$ to be the above expression with $P_n$ replaced by $Q_n$. In the isotropic case $H_n^{\mathrm{P}}$ and $H_n^{\mathrm{Q}}$ have the same spectrum, and in particular the same spectral gap.

We write $\epsilon_m^{\mathrm{T}}$ and $\epsilon_n^{\mathrm{P}}$ for the spectral gaps of $H_m^{\mathrm{T}}$ and $H_n^{\mathrm{P}}$ respectively. We prove the following 2D analogue of Theorem \ref{thm:per}.

\begin{theorem}
Let $n>2$ be even and let $m>2(n+2)$. Then
\begin{equation}
\epsilon^{\mathrm{T}}_m \geq \frac{3}{4}\left(\epsilon_n^{\mathrm{P}}-\frac{8}{n^2}\right).
\label{eq:2Dbnd}
\end{equation}
\label{thm:2D}
\end{theorem}

We have stated the Theorem for the isotropic case where the terms for horizontal and vertical edges are the same, but our proof applies with only cosmetic modification in the non isotropic case. In that case the gap $\epsilon_n^{\mathrm{P}}$ appearing on the right-hand side of \eqref{eq:2Dbnd} should be replaced by $\min\left(\epsilon_n^{\mathrm{P}},\epsilon_n^{\mathrm{Q}}\right)$ where $\epsilon_n^{\mathrm{Q}}$ is the spectral gap of $H_n^{\mathrm{Q}}$.

\subsection{ Deformed patch operators}

Since $n$ is even the patch $P_n$ has a center plaquette, shown in Figure \ref{fig:patch} for $n=6$. For any plaquette $k$ of the torus $\Lambda_m$ we define (directed) subgraphs $P_{n,k}$ and $Q_{n,k}$ of $\Lambda_m$, which are copies of $P_n$ and $Q_n$ centered at $k$. For each of these subgraphs we may consider an associated patch operator which is the sum of all terms in $H_m^{\mathrm{T}}$ which act on its edges. Below we define deformed versions of these patch operators.

It will be convenient to define a function $d(\cdot)$ on the edges of $P_n$ which measures the distance from the center plaquette in the following somewhat unconventional way.  The graph $P_n$ contains cocentered subgraphs $P_{n-2}, P_{n-4},\ldots, P_{2}$; see Figure \ref{fig:patch}. For any edge $e\in P_n$ we define the distance $d(e)$ from the center plaquette to be the smallest integer $r\in \{1,2,\ldots, \frac{n}{2}\}$ such that $e$ is contained in the cocentered patch $P_{2r}$. In Figure \ref{fig:patch}, the edges with $d(e)=1,2,3$ are shown in blue, red, and black respectively. We define $d(e)$ for edges $e\in Q_n$ in exactly the same way, i.e., with $Q$ replacing $P$ everywhere in the above paragraph.

Let $\{c_1,c_2,\ldots, c_\frac{n}{2}\}$ be a set of positive numbers and define deformed patch operators
\begin{equation}
\mathcal{B}_{n,k}=\sum_{e\in P_{n,k}} c_{d(e)}h_e \qquad \text{and} \qquad
\mathcal{C}_{n,k}=\sum_{e\in Q_{n,k}} c_{d(e)}h_e.
\label{eq:deformedpatch}
\end{equation}

These operators act on the full Hilbert space of $m^2$ qudits, but only act nontrivially on the qudits in the patches $P_{n,k}$ and $Q_{n,k}$, respectively. Looking at equation \eqref{eq:deformedpatch} we see that for each edge there is an operator multiplied by a coefficient which depends only on the distance from the edge to the center plaquette $k$ of the patch.

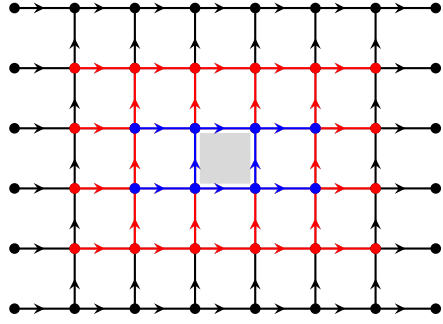
\begin{figure}[t]
\centering
\begin{tikzpicture}[scale=0.8]

\path[draw=white, fill=gray!30] (3.07,2.07) rectangle (3.93,2.93);
\foreach \x in {0,1,2,3,4,5,6}
\foreach \y in {0,1,2,3,4,5}{
\draw[thick,postaction={mid arrow}] (\x,\y)--(\x+1,\y);}
\foreach \x in {1,2,3,4,5,6}
\foreach \y in {0,1,2,3,4}
\draw[thick,postaction={mid arrow}] (\x,\y)--(\x,\y+1);

\foreach \x in {0,1,2,3,4}
\foreach \y in {0,1,2,3}{
\draw[thick,white,postaction={mid arrow}] (1+\x,1+\y)--(\x+2,\y+1);
\draw[thick,red,postaction={mid arrow}] (1+\x,1+\y)--(\x+2,\y+1);
}
\foreach \x in {1,2,3,4}
\foreach \y in {0,1,2}{
\draw[thick,white,postaction={mid arrow}] (1+\x,1+\y)--(1+\x,\y+2);
\draw[thick,red,postaction={mid arrow}] (1+\x,1+\y)--(1+\x,\y+2);
}
\foreach \x in {0,1,2,3,4,5,6,7}
\foreach \y in {0,1,2,3,4,5}
\draw[thick, fill] (\x,\y) circle (2pt);

\foreach \x in {0,1,2,3,4,5}
\foreach \y in {0,1,2,3}
\draw[thick, fill, red] (\x+1,\y+1) circle (2pt);

\foreach \x in {0,1,2}
\foreach \y in {0,1}{
\draw[thick,white,postaction={mid arrow}] (2+\x,2+\y)--(\x+3,\y+2);
\draw[thick,blue,postaction={mid arrow}] (2+\x,2+\y)--(\x+3,\y+2);
}
\foreach \x in {1,2}
\foreach \y in {0}{
\draw[thick,white,postaction={mid arrow}] (2+\x,2+\y)--(2+\x,\y+3);
\draw[thick,blue,postaction={mid arrow}] (2+\x,2+\y)--(2+\x,\y+3);
}
\foreach \x in {0,1,2,3}
\foreach \y in {0,1}{
\draw[thick, fill,blue] (\x+2,\y+2) circle (2pt);
}
\end{tikzpicture}

\caption{(Color online) The patch graph $P_6$ contains $P_4$ (in red and blue) and $P_2$ (in blue) as subgraphs which share the same center plaquette, shown in gray.  We define a distance function $d(e)$ on the edges $e$ of $P_n$, which takes integer values between $1$ and $n/2$. In this example the black edges are those with $d(e)=3$, the red edges have $d(e)=2$, and the blue edges have $d(e)=1$\label{fig:patch}.  }
\end{figure}

We will choose the coefficients $\{c_1,c_2,\ldots, c_{\frac{n}{2}}\}$ to satisfy 
\begin{align}
\text{(\textbf{Positive})} & \qquad c_j > 0 \qquad \quad &1 \leq j \leq \frac{n}{2}\label{eq:cond12D}\\
\text{(\textbf{Nonincreasing})} & \qquad c_j \leq c_{j-1} \qquad\quad  &2 \leq j \leq \frac{n}{2}.\label{eq:cond22D}
\end{align}
This is the analogue of the constraints \eqref{eq:cond1}-\eqref{eq:cond2} from the 1D case. (Here the analogue of equation \eqref{eq:cond3} is automatically enforced since, in equation \eqref{eq:deformedpatch}, the coefficient multiplying a given edge depends only on the distance to the center plaquette.)

We now show that the deformed patch operators satisfy a 2D analogue of Lemma \ref{lem:cnorm}. Consider the sets of horizontal or vertical edges in the graph $P_n$ and let
\begin{align}
\mathrm{c}_{\he} &=\bigg(\sum_{\substack{\text{horizontal}\\ e\in P_n}} c_{d(e)}\bigg)\bigg/\bigg(\sum_{\substack{\text{horizontal}\\ e\in P_n}} 1 \bigg)=\frac{1}{n(n+1)}\sum_{r=1}^{n/2}(8r-2)c_r\label{eq:chor}\\
\mathrm{c}_{\ve} &=\bigg(\sum_{\substack{\text{vertical}\\ e\in P_n}} c_{d(e)}\bigg)\bigg/\bigg(\sum_{\substack{\text{vertical}\\ e\in P_n}} 1 \bigg)=\frac{1}{n(n-1)}\sum_{r=1}^{n/2}(8r-6)c_r\label{eq:cvert}
\end{align}
be the corresponding averages of horizontal-edge and vertical-edge coefficients (the first equality in each of the above lines is a definition, while the second equality is the result of a simple computation).

\begin{lemma}
Consider the eigenspace of $H^{\mathrm{T}}_m$ with eigenvalue $\epsilon^{\mathrm{T}}_m$. There exists a normalized state $|\phi\rangle$ in this eigenspace which satisfies
\begin{equation}
\langle \phi |\mathcal{B}_{n,k}^2 |\phi\rangle \geq \mathrm{min}\big(\mathrm{c}_{\he}\; , \; \mathrm{c}_{\ve}\big)\langle \phi |\mathcal{B}_{n,k}|\phi\rangle \epsilon^{\mathrm{P}}_n
\label{eq:2Dbnd1}
\end{equation}
and
\begin{equation}
\langle \phi |\mathcal{C}_{n,k}^2 |\phi\rangle \geq  \mathrm{min}\big(\mathrm{c}_{\he}\; , \; \mathrm{c}_{\ve}\big) \langle \phi |\mathcal{C}_{n,k}|\phi\rangle \epsilon^{\mathrm{P}}_n
\label{eq:2Dbnd2}
\end{equation}
for each plaquette $k$ in the $m\times m$ torus $\Lambda_m$.
\label{lem:ab2}
\end{lemma}
\begin{proof}
The proof follows that of Lemma \ref{lem:cnorm} very closely. To avoid repetition, we only describe the (small) differences here. On the torus there are two spatial translation generators, unitary operators which shift all qudits by one site in the horizontal and vertical directions, respectively. These operators and the Hamiltonian $H_m^{\mathrm{T}}$ all mutually commute and we choose $|\phi\rangle$ to be a simultaneous eigenvector of all of them. To show that $|\phi\rangle$ satisfies \eqref{eq:2Dbnd1} one first establishes the analog of \eqref{eq:Bsq}:
\begin{equation}
\langle \phi |\mathcal{B}_{n,k}^2|\phi\rangle\geq \langle \widehat{\phi}|\mathcal{B}_{n,k} |\widehat{\phi}\rangle \langle \phi|\mathcal{B}_{n,k} |\phi\rangle
\label{eq:phiB}
\end{equation}
(where $|\widehat{\phi}\rangle$ is defined by \eqref{eq:phihat} with $G_{n,k}$ the null space of $\mathcal{B}_{n,k}$). Using the fact that $|\phi\rangle$ is an eigenstate of both spatial translation operators, one shows (by the same logic as in \eqref{eq:sym}) that  $\langle \widehat{\phi}|h_e|\widehat{\phi}\rangle$ takes the same value for all horizontal edges in $P_{n,k}$ and likewise for all vertical edges in $P_{n,k}$ (although the two values may be different). Using this fact we have
\begin{align}
\langle \widehat{\phi}|\mathcal{B}_{n,k} |\widehat{\phi}\rangle &=\mathrm{c}_{\he}\sum_{\substack{\text{horizontal}\\ e\in P_{n,k}}}\langle \widehat{\phi}|h_e|\widehat{\phi}\rangle+\mathrm{c}_{\ve}\sum_{\substack{\text{vertical}\\ e\in P_{n,k}}} \langle \widehat{\phi}|h_e|\widehat{\phi}\rangle\label{eq:q1}\\
&\geq \mathrm{min}\big(\mathrm{c}_{\he}\; , \; \mathrm{c}_{\ve}\big) \sum_{e\in P_{n,k}} \langle \widehat{\phi}|h_e|\widehat{\phi}\rangle\label{eq:q2}\\
&\geq \mathrm{min}\big(\mathrm{c}_{\he}\; , \; \mathrm{c}_{\ve}\big) \epsilon_n^{\mathrm{P}}.\label{eq:q3}
\end{align}
To go from \eqref{eq:q2} to \eqref{eq:q3} we use the fact that $|\widehat{\phi}\rangle$ is a normalized state orthogonal to the nullspace of the operator $\sum_{e\in P_{n,k}}h_e$ and therefore has energy lower bounded by its smallest non zero eigenvalue $\epsilon_n^{\mathrm{P}}$. Plugging \eqref{eq:q3} into \eqref{eq:phiB} gives the desired bound \eqref{eq:2Dbnd1}. The proof of \eqref{eq:2Dbnd2} follows exactly the same steps but with $\mathcal{C}_{n,k}$ instead of $\mathcal{B}_{n,k}$.
\end{proof}

\subsection{Proof of Theorem \ref{thm:2D}}
\begin{proof}
We expand the squares:
\begin{equation}
(H_m^{\mathrm{T}})^2=\sum_{e\in \Lambda_m}h_{e}+\sum_{\substack{e_1,e_2\in \Lambda_m \\ e_1\neq e_2}}h_{e_1}h_{e_2}
\label{eq:hm2}
\end{equation}
and
\begin{equation}
\sum_{\substack{\mathrm{plaquettes}\\k\in \Lambda_m}} \left(\mathcal{B}_{n,k}^2+ \mathcal{C}_{n,k}^2\right)=\sum_{e\in \Lambda_m}W_n(e,e)h_{e}+\sum_{\substack{e_1,e_2\in \Lambda_m \\ e_1\neq e_2}}W_n(e_1,e_2)h_{e_1}h_{e_2}.
\label{eq:absq}
\end{equation}
Here $W_n(\cdot,\cdot)$ is a real-valued function on pairs of edges of $\Lambda_m$, which depends on the coefficients $\{c_1,c_2,\ldots, c_{\frac{n}{2}}\}$. The value $W_n(e_1,e_2)$ includes contributions from all terms $\mathcal{B}_{n,k}^2$ on the left-hand side of \eqref{eq:absq} such that $e_1,e_2\in P_{n,k}$ and from all terms $\mathcal{C}_{n,k}^2$ such that $e_1,e_2\in Q_{n,k}$.  To obtain an explicit expression, define an equivalence relation $\sim$ on pairs of edges of $\Lambda_m$, which identifies pairs $(e_1,e_2) \sim (T(e_1),T(e_2))$ that differ only by a lattice translation $T$. We may then write $W_n$ as a sum of two autocorrelation functions:
\begin{equation}
W_n(e_1,e_2)=\sum_{\substack{(e, e')\sim (e_1,e_2)\\ e,e'\in P_{n,k}}} c_{d(e)}c_{d(e')}+\sum_{\substack{(e, e')\sim (e_1,e_2)\\ e,e'\in Q_{n,k}}} c_{d(e)}c_{d(e')}.
\label{eq:Wn}
\end{equation}
The right-hand side does not depend on the choice of center plaquette $k$ used to evaluate it. Note that \eqref{eq:Wn} shows that $W_n(\cdot,\cdot)$ is invariant if we rotate both of its arguments by $\pi/2$. Also note that our assumption that $m>2(n+2)$ implies that a subgraph $P_{n,k}$ or $Q_{n,k}$ extends less than halfway around the torus $\Lambda_m$ in each direction. This implies, e.g., that for any pair of edges $e,e' \in P_{n,k}$ which contribute to the sum in the first term, both the horizontal and vertical distances (number of edges) between them on the torus $\Lambda_m$ is the same as those distances within the patch $P_{n,k}$ (cf. the observation in parentheses after equation \eqref{eq:sumsq}).  It will be convenient to introduce notation for the following function values:
\begin{align}
W_n^{\he}&=W_n(e,e) \\  
W_n^{\hee}&=W_n(e_1,e_2) \quad \parbox{25em}{\vphantom{a} where $e_1,e_2$ are collinear (both horizontal or both vertical) and share a vertex. \\} \label{eq:hee}\\
W_n^{\hve}&=W_n(e_1,e_2)  \quad \text{where $e_1,e_2$ are not collinear and share a vertex.} \label{eq:hee}
\end{align}
We now evaluate these quantities using equation \eqref{eq:Wn}. We have
\begin{equation}
W_n^{\he}=\sum_{e\in P_n} (c_{d(e)})^2=\sum_{r=1}^{n/2} (16r-8)c_r^2
\label{eq:we}
\end{equation}
and
\begin{align}
W_n^{\hee}&=W_{n-2}^{\hee}+4(n-2)c_{n/2}c_{n/2-1}+(4n-4)c_{n/2}^2\label{eq:recur}\\
&=\sum_{r=1}^{n/2}(8r-4)c_r^2+\sum_{r=2}^{n/2} (8r-8)c_rc_{r-1}.
\label{eq:wee}
\end{align}
A straightforward calculation shows that $W_1^{\hve}=W_1^{\hee}=4c_1^2$ and that $W_n^{\hve}$ satisfies the same recursion \eqref{eq:recur} as $W_n^{\hee}$. Therefore
\begin{equation}
W_n^{\hve}=W_n^{\hee}.
\label{eq:equality}
\end{equation}
This equality is not a coincidence--we carefully chose the shape of the patch $P_n$ to make it happen. It is essential for the next step of the proof.

Now define
\begin{equation}
\alpha=\frac{1}{W_n^{\hee}}\qquad \qquad \beta=\alpha\left(W_n^{\he}-W_n^{\hee}\right).
\label{eq:alph}
\end{equation}
Using equations \eqref{eq:hm2}, \eqref{eq:absq} we get
\begin{equation}
(H_m^{\mathrm{T}})^2-\alpha \sum_{\substack{\mathrm{plaquettes}\\k\in \Lambda_m}} \left(\mathcal{B}_{n,k}^2+ \mathcal{C}_{n,k}^2\right)+\beta H_m^{\mathrm{T}}=\sum_{\substack{e_1,e_2\in \Lambda_m \\ e_1\neq e_2}}(1-\alpha W_n(e_1,e_2))h_{e_1}h_{e_2}.
\label{eq:bigsum}
\end{equation}
We now show that each term in the sum on the right-hand side is a positive semidefinite operator. First consider a term where the edges $e_1$ and $e_2$ share a vertex. They may be collinear or not; in either case we have $W_n(e_1,e_2)=W_n^{\hee}=W_n^{\hve}=\alpha^{-1}$ and so the corresponding term in \eqref{eq:bigsum} vanishes. Next suppose that $e_1,e_2$ do not share a vertex. In this case the operators $h_{e_1}$ and $h_{e_2}$ are commuting projectors and therefore  $h_{e_1}h_{e_2}\geq 0$. Furthermore, the following Lemma states that the coefficient $(1-\alpha W_n(e_1,e_2))$ is also non-negative, implying $(1-\alpha W_n(e_1,e_2))h_{e_1}h_{e_2}\geq 0$.
\begin{autocor2}
For any two edges $e_1,e_2 \in \Lambda_m$ which do not share a vertex, we have $W_n(e_1,e_2)\leq W_n^{\hee}$
\label{lem:autocor2}
\end{autocor2}
The proof, given in the Appendix, uses the representation \eqref{eq:Wn} of $W_n(e_1,e_2)$ as a sum of two autocorrelation functions and shows that they satisfy an analogue of the monotonicity property stated in the \aone.

Thus each term on the right-hand side of \eqref{eq:bigsum} is positive semidefinite and therefore
\[
(H_m^{\mathrm{T}})^2+ \beta H_m^{\mathrm{T}}\geq \alpha \sum_{\substack{\mathrm{plaquettes}\\k\in \Lambda_m}} \left(\mathcal{B}_{n,k}^2+\mathcal{C}_{n,k}^2\right).
\]
Taking the expectation value of this equation in the state $|\phi\rangle$ from Lemma \ref{lem:ab2}, applying the Lemma, and using the fact that
\[
\sum_{\substack{\mathrm{plaquettes}\\k\in \Lambda_m}} \left(\mathcal{B}_{n,k}+ \mathcal{C}_{n,k}\right)=\left(\sum_{e\in P_n} c_{d(e)}\right) H_m^{\mathrm{T}},
\]
we get
\[
(\epsilon_m^{\mathrm{T}})^2+\beta \epsilon_m^{\mathrm{T}} \geq \alpha \epsilon_n^{\mathrm{P}} \epsilon_m^{\mathrm{T}} \left(\sum_{e\in P_n} c_{d(e)}\right)\mathrm{min}\big(\mathrm{c}_{\he}\; , \; \mathrm{c}_{\ve}\big).
\]
Dividing through by $\epsilon_m^{\mathrm{T}}$, rearranging, and using \eqref{eq:alph} gives
\begin{equation}
\epsilon_m^{\mathrm{T}}\geq f(n)\left(\epsilon_n^{\mathrm{P}}-g(n)\right)
\label{eq:emt}
\end{equation}
with
\begin{equation}
f(n)=\frac{\mathrm{min}\big(\mathrm{c}_{\he}\; , \; \mathrm{c}_{\ve}\big)\sum_{e\in P_n} c_{d(e)}}{W_n^{\hee}}
\label{eq:fn}
\end{equation}
and
\begin{equation}
g(n)=\frac{W_n^{\he}-W_n^{\hee}}{\mathrm{min}\big(\mathrm{c}_{\he}\; , \; \mathrm{c}_{\ve}\big)\sum_{e\in P_n} c_{d(e)}}.
\label{eq:gn}
\end{equation}
Note that
\begin{equation}
\sum_{e\in P_n} c_{d(e)}=\sum_{r=1}^{n/2} (16r-8)c_r.
\label{eq:sumc}
\end{equation}
With this in hand, we now have explicit expressions (as finite sums) for all quantities appearing in equations \eqref{eq:fn} and \eqref{eq:gn}, given in equations \eqref{eq:chor}, \eqref{eq:cvert}, \eqref{eq:we}, \eqref{eq:wee}, and \eqref{eq:sumc}. Given a choice of coefficients $\{c_1,c_2,\ldots, c_{n/2}\}$ we can evaluate $f(n)$ and $g(n)$ using these expressions. We choose
\[
c_j=\frac{n}{2}\left(\frac{n}{2}+1\right)-j(j-1)  \qquad j=1,\ldots ,n/2
\]
(which, as required, is positive and nonincreasing). With this choice all the finite sums can be computed exactly (we omit the details) and we obtain 
\begin{align}
f(n)&=\frac{3}{4}\left(\frac{n+2}{n-1}\right)\frac{n^2+\frac{2}{3}n-\frac{4}{3}}{n^2+2n-2}\geq \frac{3}{4}\\
g(n)&=\left(\frac{n-1}{n+2}\right)\frac{8}{n^2+\frac{2}{3}n-\frac{4}{3}}\leq \frac{8}{n^2}.
\end{align}
Plugging these bounds into \eqref{eq:emt} completes the proof.
\end{proof}
\section{Acknowledgments}
We thank Fernando Brand\~{a}o, Yichen Huang, Alexei Kitaev, Bruno Nachtergaele, Daniel Nagaj, and John Preskill for helpful comments. We also thank Alexei Kitaev for sharing his calculation \cite{Kit} with us. We acknowledge funding provided by the Institute for Quantum Information and Matter, an NSF Physics Frontiers Center (NFS Grant PHY-1125565) with support of the Gordon and Betty Moore Foundation (GBMF-12500028).
\bibliographystyle{plain}
\bibliography{Knabe}
\begin{appendix}
\section*{Appendix}
In this appendix we prove the \aone and the \atwo, which establish monotonicity properties of certain autocorrelation functions in one and two dimensions respectively.

\begin{proof}[Proof of the \aone]
Let $\{c_j: j\in\{0,1\ldots,n-2\}\}$ satisfy \eqref{eq:cond1}-\eqref{eq:cond3}. First extend $c$ to a piecewise constant function on the real line as follows:
\[
c(x)=\begin{cases}
c_{\lfloor x\rfloor} & x\in [0, \frac{n-2}{2}]\\
c_{\lceil x\rceil} & x\in (\frac{n-2}{2}, n-2]\\
0 & x\in (-\infty, 0)\cup (n-2,\infty).
\end{cases}
\]
where $\lfloor x\rfloor$ and  $\lceil x\rceil$ indicate, respectively, the largest integer less than or equal to $x$ and the smallest integer larger than $x$. As an example, consider the choice  (with $n=8$)
\[
\{c_0,c_1,c_2,c_3,c_4,c_5,c_6\}=\{1,3,3,4.25,3,3,1\}.
\]
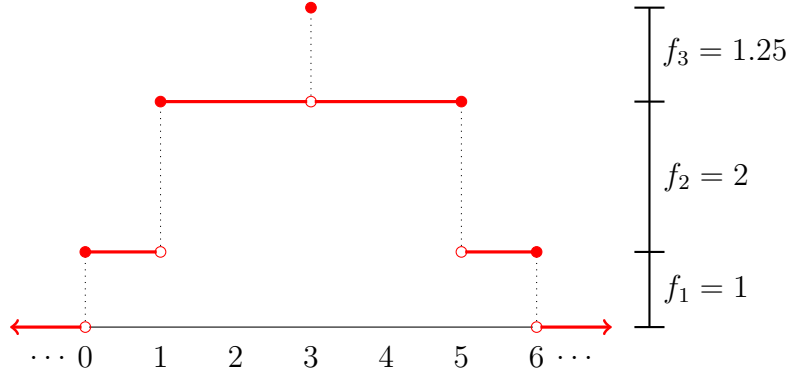
\begin{figure}[t]
\centering
\begin{tikzpicture}
\draw (-2.95,0)--(2.95,0);

\draw [dotted] (-3,0.08)--(-3,1);
\draw [dotted] (-2,1.08)--(-2,3);
\draw [very thick, red] (-3,1)--(-2.05,1);
\draw [very thick, red] (-2,3)--(-0.05,3);
\draw [dotted] (3,0.08)--(3,1);
\draw [dotted] (2,1.08)--(2,3);
\draw [very thick,red] (2.05,1)--(3,1);
\draw [very thick, red] (0.05,3)--(2,3);
\draw [dotted] (0,3.08)--(0,4.25);

\draw [->,very thick, red] (-3.05,0)--(-4,0);
\draw [->,very thick, red] (3.05,0)--(4,0);

\draw[solid, red, fill] (0,4.25) circle (2pt) ;
\draw[solid, red, fill] (-2,3) circle (2pt) ;

\draw[solid, red, fill] (2,3) circle (2pt) ;
\draw[solid, red, fill] (-3,1) circle (2pt) ;
\draw[solid, red, fill] (3,1) circle (2pt) ;

\draw[solid, red] (0,3) circle (2pt) ;
\draw[solid, red] (-2,1) circle (2pt) ;
\draw[solid, red] (2,1) circle (2pt) ;

\draw[solid, red] (3,0) circle (2pt) ;

\draw[solid, red] (-3,0) circle (2pt) ;

\draw (-3.5,-0.4) node {$\ldots$};
\draw (-3,-0.4) node {0};
\draw (-2,-0.4) node {1};
\draw (-1,-0.4) node {2};
\draw (0,-0.4) node {3};
\draw (1,-0.4) node {4};
\draw (2,-0.4) node {5};
\draw (3,-0.4) node {6};
\draw (3.5,-0.4) node {$\ldots$};

\draw [thick] (4.5, 0)--(4.5, 4.25);
\draw [thick] (4.3, 0)--(4.7, 0);
\draw [thick] (4.3, 1)--(4.7, 1);
\draw [thick] (4.3, 3)--(4.7, 3);
\draw [thick] (4.3, 4.25)--(4.7, 4.25);

\draw (5.25,0.5) node {$f_1=1$};

\draw (5.25,2) node {$f_2=2$};

\draw (5.5,3.66) node {$f_3=1.25$};
\end{tikzpicture}
\caption{We extend $\{c_1,\ldots, c_{n-2}\}$ to a function $c(x)$ on $\mathbb{R}$. Here we show an example with $n=8$ and $\{c_0,c_1,c_2,c_3,c_4,c_5,c_6\}=\{1,3,3,4.25,3,3,1\}$. \label{fig:piece}}
\end{figure}

In this example the function $c(x)$ is plotted in red in Figure \ref{fig:piece}. It is easy to see that any such piecewise constant function can be decomposed as a weighted sum of at most $n/2$ indicator functions which are each symmetric about the midpoint $\frac{n-2}{2}$:
\begin{equation}
c(x)=\sum_{k} f_k a_k(x) \qquad 
\label{eq:ind}
\end{equation}
where $f_k> 0$ are positive coefficients and
\[
a_k(x)=\begin{cases}
1 & k\in I_k\\
0 & \text{otherwise}.
\end{cases}
\]
Here each $I_k\subseteq [0,n-2]$ is a closed interval with midpoint $\frac{n-2}{2}$. In the example considered above we may take $I_1=[0,6]$, $I_2=[1,5]$, $I_3=[3,3]$ and $f_1=1$, $f_2=2$, $f_3=1.25$.

Now using \eqref{eq:ind} we obtain an expression for the autocorrelation:
\[
q(x)=\sum_{j=0}^{n-x-2}c_jc_{j+x}=\sum_{j=0}^{n-x-2}c(j)c(j+x)=\sum_{k,k'} f_k f_{k'}\left(\sum_{j=0}^{n-x-2} a_k(j)a_{k'}(j+x)\right).
\]
To prove the Lemma, it suffices to show that, for each pair $k,k'$, the expression in parentheses is nonincreasing as a function of $x\in [0,\infty)$. Using the fact that $a_k(\cdot)$ is an indicator function associated with the interval $I_k$ we have
\begin{equation}
\sum_{j=0}^{n-2-x} a_k(j)a_{k'}(j+x)=\left|I_k\cap (I_{k'}-x)\right|
\label{eq:indshift}
\end{equation}
where $I_{k'}-x$ is defined to be the interval $I_k$ shifted to the left by $x$. To complete the proof note that when $x=0$ the intervals $I_k$ and $I_{k'}-x$ are cocentered and the size of their intersection is equal to the minimum of their sizes, i.e., $\min\{|I_k|,|I_{k'}|\}$, and that as $x$ increases this intersection can only decrease.
\end{proof}

Next we prove the \atwo. Here the details are different but the strategy is essentially the same as in the proof given above. 

\begin{proof}[Proof of the \atwo]
Let $e_1,e_2 \in \Lambda_m$ be two edges which do not share a vertex. We use the representation \eqref{eq:Wn} of $W_n(e_1,e_2)$ as a sum of two autocorrelation functions. Denote these two functions, the first and second terms in \eqref{eq:Wn}, by $A_1(e_1,e_2)$ and $A_2(e_1,e_2)$. To prove the Lemma it is sufficient to show that there exists an edge $e_3$ which shares a vertex with $e_1$ such that
\begin{equation}
A_1(e_1,e_2)\leq A_1(e_1,e_3)
\label{eq:mon}
\end{equation}
and
\begin{equation}
A_2(e_1,e_2)\leq A_2(e_1,e_3).
\label{eq:mon2}
\end{equation}
In particular this shows that $W_n(e_1,e_2)\leq W_n(e_1,e_3)$, where the right-hand side is either $W_n^{\hee}$ or $W_n^{\hve}$ (by definition, since $e_1,e_3$ share a vertex). The Lemma then follows from the fact that $W_n^{\hve}=W_n^{\hee}$ (equation \eqref{eq:equality}). 

Below we provide a proof of \eqref{eq:mon}; the proof of \eqref{eq:mon2} is almost identical. We consider
\[
A_1(e_1,e_2)=\sum_{\substack{(e, e')\sim (e_1,e_2)\\ e,e'\in P_{n,k}}} c_{d(e)}c_{d(e')}.
\]
The right-hand side does not depend on the plaquette $k$ in $\Lambda_m$ but for concreteness we imagine fixing some specific choice in the following.  This function takes the same value for all pairs of edges in the equivalence class of $(e_1,e_2)$ (under the relation $\sim$). Moreoever if the function value $A_1(e_1,e_2)$ is non zero then there are a pair of edges of $P_{n,k}$ in this equivalence class.  For this reason, without loss of generality, below we assume $e_1,e_2$ are edges of $P_{n,k}$.

It will be helpful to imagine embedding the graph $P_{n,k}$ in $\mathbb{R}^2$, taking each edge to have length $1$. Define $T_{e_1\rightarrow e_2}$ to be the linear translation of $\mathbb{R}^2$ which takes the midpoint of edge $e_1$ into the midpoint of edge $e_2$\footnote{In other words, for $x\in \mathbb{R}^2$ we have $T_{e_1\rightarrow e_2}(x)=x+e_2-e_1$ (here $e_2,e_1$ are identified with the coordinates of their midpoints).}.

Let us now fix the edge $e_3$ for which we claim \eqref{eq:mon} (and \eqref{eq:mon2}) holds. We choose $e_3\in P_{n,k}$ to be the unique edge with the following properties:
\begin{enumerate}
\item $e_3$ and $e_1$ share a vertex
\item $e_3$ and $e_2$ have the same orientation (vertical or horizontal) 
\item $e_3$ is closer to $e_2$ than $e_1$ is, i.e., the distance (in $\mathbb{R}^2$) between the midpoints of $e_3$ and $e_2$ is smaller than the corresponding distance between the midpoints of $e_1$ and $e_2$. 
\end{enumerate}
Let $T_{e_1\rightarrow e_3}$ be the translation which maps the midpoint of $e_1$ into the midpoint of $e_3$. 

Following the proof of the \aone, it will be helpful to decompose the function $c_{d(e)}$ as a weighted sum of indicator functions. Using the fact that $c_1,c_2,\ldots, c_{n/2}$ is positive and nonincreasing we may write 
\[
c_{d(e)}=\sum_{j=1}^{n/2}  f_ja_j(e)
\]
where $f_j\geq 0$ and 
\[
a_j(e)=\begin{cases}
1 & d(e)\leq j\\
0 & \text{otherwise}.
\end{cases} \qquad e\in P_{n,k}
\]
Then 
\begin{equation}
A_1(e_1,e_2)=\sum_{j,j'=1}^{n/2} f_j f_{j'}R_{j,j'}(e_1,e_2)
\label{eq:A1}
\end{equation}
where
\[
R_{j,j'}(e_1,e_2)=\sum_{\substack{(e, e')\sim (e_1,e_2)\\ e,e'\in P_{n,k}}} a_j(e)a_{j'}(e').
\]
We now show that
\begin{equation}
R_{j,j'}(e_1,e_2)\leq R_{j,j'}(e_1,e_3)
\label{eq:Rmon}
\end{equation}
which along with \eqref{eq:A1} implies \eqref{eq:mon}.

To see why \eqref{eq:Rmon} holds, first note that the function $R_{j',j}(e_1,e_2)$ can be computed as follows.  Consider overlapping and cocentered graphs $P_{2j}$ and $P_{2j'}$, both embedded in $\mathbb{R}^2$ in the manner described above. Form new graphs $\tilde{P}_{2j}$ and $\tilde{P}_{2j'}$ by discarding all edges of $P_{2j}$ which are not parallel to $e_1$ and all edges of $P_{2j'}$ which are not parallel to $e_2$. Then apply the translation $T_{e_1\rightarrow e_2}$ to $\tilde{P}_{2j}$. The number of edge midpoints which overlap between $\tilde{P}_{2j'}$ and this translated version of $\tilde{P}_{2j}$ is equal to the function value $R_{j,j'}(e_1,e_2)$. Let us write this as
\begin{equation}
R_{j,j'}(e_1,e_2)=\mathrm{EM}\left[\tilde{P}_{2j'}, T_{e_1\rightarrow e_2}(\tilde{P}_{2j})\right]
\label{eq:r1}
\end{equation}
where $\mathrm{EM}\left[\cdot,\cdot\right]$ counts the number of edge midpoints which coincide. Likewise we have
\begin{equation}
R_{j,j'}(e_1,e_3)=\mathrm{EM}\left[\tilde{P}_{2j'}, T_{e_1\rightarrow e_3}(\tilde{P}_{2j})\right]
\label{eq:r2}
\end{equation}

The translation $T_{e_1\rightarrow e_2}$ can be decomposed uniquely as 
\begin{equation}
T_{e_1\rightarrow e_2}=T_v T_h T_{e_1\rightarrow e_3},
\label{eq:Tdecomp}
\end{equation}
where $T_v$ is a vertical translation (i.e., $T_v(x)=x+(0,a)$ for some integer $a$) and $T_h$ is a horizontal translation ($T_h(x)=x+(b,0)$ for integer $b$). Recall that we chose $e_3$ so that it is closer to $e_2$ than $e_1$ is. Therefore, if it is nontrivial, the translation $T_h$ (resp. $T_v$) increases the difference between the horizontal (resp. vertical) coordinates of the center plaquettes of the two graphs. As a result it is not hard to see that each of these translations can only decrease the number of edge midpoints which overlap, i.e., 
\begin{equation}
 \mathrm{EM}\left[\tilde{P}_{2j'}, T_{e_1\rightarrow e_3}(\tilde{P}_{2j})\right]\geq \mathrm{EM}\left[\tilde{P}_{2j'},T_hT_{e_1\rightarrow e_3}(\tilde{P}_{2j})\right]\geq \mathrm{EM}\left[\tilde{P}_{2j'}, T_vT_hT_{e_1\rightarrow e_3}(\tilde{P}_{2j})\right].
\label{eq:em}
\end{equation}

It may help to refer to Figure \ref{fig:shift} which illustrates an example.  Equation \eqref{eq:Rmon} then follows directly from \eqref{eq:r1},\eqref{eq:r2}, \eqref{eq:Tdecomp}, and \eqref{eq:em}.

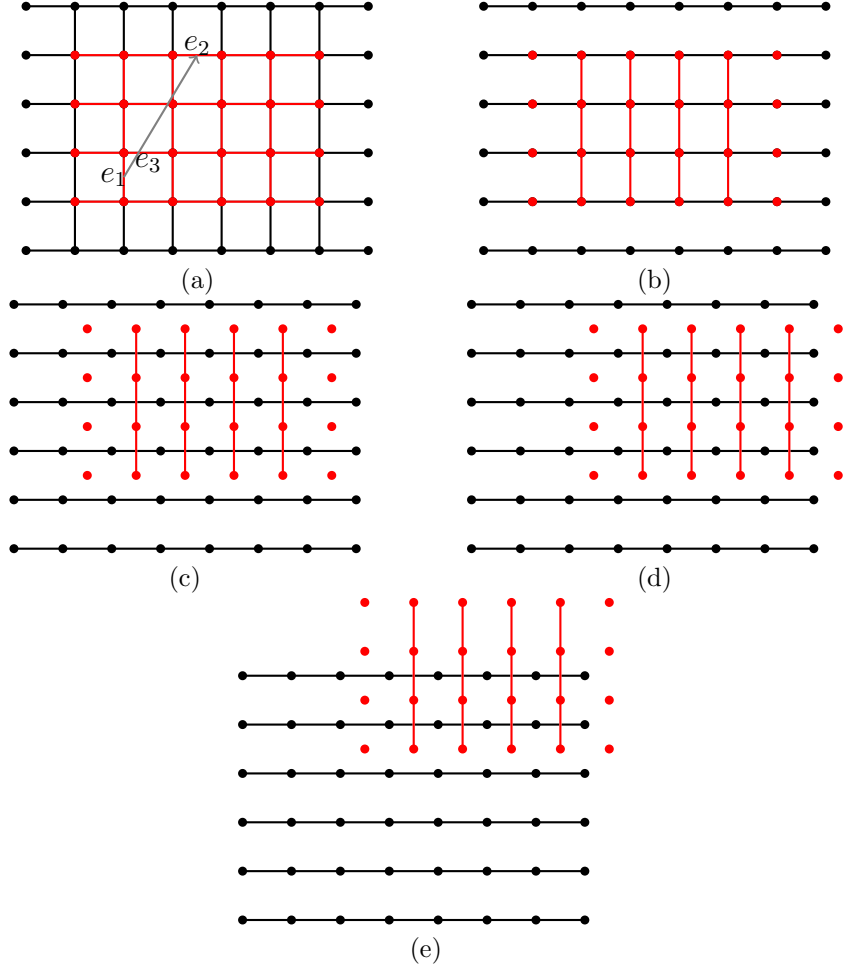
\begin{figure}[t]
\centering
\subfloat[]{
\begin{tikzpicture}[scale=0.65]
\foreach \x in {0,1,2,3,4,5,6}
\foreach \y in {0,1,2,3,4,5}{
\draw[thick] (\x,\y)--(\x+1,\y);}
\foreach \x in {1,2,3,4,5,6}
\foreach \y in {0,1,2,3,4}
\draw[thick] (\x,\y)--(\x,\y+1);
\foreach \x in {0,1,2,3,4}
\foreach \y in {0,1,2,3}{
\draw[thick,white] (1+\x,1+\y)--(\x+2,\y+1);
\draw[thick,red] (1+\x,1+\y)--(\x+2,\y+1);
}
\foreach \x in {1,2,3,4}
\foreach \y in {0,1,2}{
\draw[thick,white] (1+\x,1+\y)--(1+\x,\y+2);
\draw[thick,red] (1+\x,1+\y)--(1+\x,\y+2);
}
\foreach \x in {0,1,2,3,4,5,6,7}
\foreach \y in {0,1,2,3,4,5}
\draw[thick, fill] (\x,\y) circle (2pt);

\foreach \x in {0,1,2,3,4,5}
\foreach \y in {0,1,2,3}
\draw[thick, fill, red] (\x+1,\y+1) circle (2pt);

\draw[thick,gray,->] (2,1.5)--(3.5,4);
\draw (1.8,1.5) node[]{$e_1$};
\draw (3.5,4.2) node[]{$e_2$};
\draw (2.5,1.8) node[]{$e_3$};

\end{tikzpicture}
}
\hspace{1cm}
\subfloat[]{
\begin{tikzpicture}[scale=0.65]
\foreach \x in {0,1,2,3,4,5,6}
\foreach \y in {0,1,2,3,4,5}{
\draw[thick] (\x,\y)--(\x+1,\y);}

\foreach \x in {1,2,3,4}
\foreach \y in {0,1,2}{
\draw[thick,white] (1+\x,1+\y)--(1+\x,\y+2);
\draw[thick,red] (1+\x,1+\y)--(1+\x,\y+2);
}
\foreach \x in {0,1,2,3,4,5,6,7}
\foreach \y in {0,1,2,3,4,5}
\draw[thick, fill] (\x,\y) circle (2pt);

\foreach \x in {0,1,2,3,4,5}
\foreach \y in {0,1,2,3}
\draw[thick, fill, red] (\x+1,\y+1) circle (2pt);

\end{tikzpicture}
}
\hspace{1cm}
\subfloat[]{
\begin{tikzpicture}[scale=0.65]
\foreach \x in {0,1,2,3,4,5,6}
\foreach \y in {0,1,2,3,4,5}{
\draw[thick] (\x,\y)--(\x+1,\y);}

\foreach \x in {0,1,2,3,4,5,6,7}
\foreach \y in {0,1,2,3,4,5}
\draw[thick, fill] (\x,\y) circle (2pt);

\begin{scope}[xshift=0.5cm,yshift=0.5cm]

\foreach \x in {1,2,3,4}
\foreach \y in {0,1,2}{
\draw[thick,white] (1+\x,1+\y)--(1+\x,\y+2);
\draw[thick,red] (1+\x,1+\y)--(1+\x,\y+2);
}

\foreach \x in {0,1,2,3,4,5}
\foreach \y in {0,1,2,3}
\draw[thick, fill, red] (\x+1,\y+1) circle (2pt);
\end{scope}
\end{tikzpicture}
}
\hspace{1cm}
\subfloat[]{
\begin{tikzpicture}[scale=0.65]
\foreach \x in {0,1,2,3,4,5,6}
\foreach \y in {0,1,2,3,4,5}{
\draw[thick] (\x,\y)--(\x+1,\y);}

\foreach \x in {0,1,2,3,4,5,6,7}
\foreach \y in {0,1,2,3,4,5}
\draw[thick, fill] (\x,\y) circle (2pt);

\begin{scope}[xshift=1.5cm,yshift=0.5cm]

\foreach \x in {1,2,3,4}
\foreach \y in {0,1,2}{
\draw[thick,white] (1+\x,1+\y)--(1+\x,\y+2);
\draw[thick,red] (1+\x,1+\y)--(1+\x,\y+2);
}

\foreach \x in {0,1,2,3,4,5}
\foreach \y in {0,1,2,3}
\draw[thick, fill, red] (\x+1,\y+1) circle (2pt);
\end{scope}
\end{tikzpicture}
}
\hspace{1cm}
\subfloat[]{
\begin{tikzpicture}[scale=0.65]
\foreach \x in {0,1,2,3,4,5,6}
\foreach \y in {0,1,2,3,4,5}{
\draw[thick] (\x,\y)--(\x+1,\y);}

\foreach \x in {0,1,2,3,4,5,6,7}
\foreach \y in {0,1,2,3,4,5}
\draw[thick, fill] (\x,\y) circle (2pt);

\begin{scope}[xshift=1.5cm,yshift=2.5cm]

\foreach \x in {1,2,3,4}
\foreach \y in {0,1,2}{
\draw[thick,white] (1+\x,1+\y)--(1+\x,\y+2);
\draw[thick,red] (1+\x,1+\y)--(1+\x,\y+2);
}

\foreach \x in {0,1,2,3,4,5}
\foreach \y in {0,1,2,3}
\draw[thick, fill, red] (\x+1,\y+1) circle (2pt);
\end{scope}
\end{tikzpicture}
}
\caption{(Color online) Here we take $n=3$, $j=2$ and $j'=3$ and we illustrate how to compute $R_{j,j'}(e_1,e_2)$, where the edges $e_1,e_2$, and $e_3$ are shown in (a). We start with overlapping cocentered copies of $P_{2j}$ (red in (a)) and $P_{2j'}$ (red and black in (a)). Although these are directed graphs we have not drawn the edge orientations since they are irrelevant here. Form new graphs $\tilde{P}_{2j}$ and $\tilde{P}_{2j'}$, shown in (b), by retaining only the edges which are parallel to $e_1$ and $e_2$, respectively.  Then apply the shift $T_{e_1\rightarrow e_2}$ to $\tilde{P}_{2j}$ (shown in (e)); then $R_{j,j'}(e_1,e_2)=8$ is the number of edge midpoints which overlap between this graph and $\tilde{P}_{2j'}$. We can decompose the shift as $T_{e_1\rightarrow e_2}=T_v T_h T_{e_1\rightarrow e_3}$ where $T_v$ and $T_h$ are vertical and horizontal translations respectively. Subfigures (c) and (d) show $\tilde{P}_{2j'}$ along with  $T_{e_1\rightarrow e_3}(\tilde{P}_{2j})$ and $T_hT_{e_1\rightarrow e_3}(\tilde{P}_{2j})$, respectively. \label{fig:shift}}
\end{figure}

\end{proof}
\end{appendix}
\end{document}